\documentclass[conference]{IEEEtran}
\usepackage{cite}

\ifCLASSINFOpdf
  \usepackage[pdftex]{graphicx}
\else
\fi
\usepackage{color}
\usepackage{caption}
\usepackage{subcaption}

\usepackage[cmex10]{amsmath}
\usepackage{amsthm}
\usepackage{amssymb}

\usepackage{algorithmic}

\usepackage{tikz}
\usetikzlibrary{arrows,decorations.pathreplacing}
\usepackage{pgfplots}

\usepackage{stfloats}

\usepackage{threeparttable}

\newtheorem{proposition}{Proposition}

\IEEEoverridecommandlockouts

\newcommand{\sfL}{\mathsf{L}}

\newcommand{\calA}{\mathcal{A}}

\newcommand{\LL}[2]{ {\sfL_{#1}^{(#2)}} }
\newcommand{\uu}[2]{ {u_{#1}^{(#2)}} }

\begin{document}
%
\title{A Low-Complexity Improved Successive Cancellation Decoder for Polar Codes}
%
%
%

\author{\authorblockA{Telecommunications Circuits Laboratory, EPFL, Lausanne, Switzerland.}}%
\author{\authorblockN{Orion Afisiadis, Alexios Balatsoukas-Stimming, and Andreas Burg}%
\authorblockA{Telecommunications Circuits Laboratory, \'Ecole Polytechnique F\'ed\'erale de Lausanne, Switzerland.}}%

%
%

\maketitle

\begin{abstract}
Under successive cancellation (SC) decoding, polar codes are inferior to other codes of similar blocklength in terms of frame error rate. While more sophisticated decoding algorithms such as list- or stack-decoding partially mitigate this performance loss, they suffer from an increase in complexity. In this paper, we describe a new flavor of the SC decoder, called the \emph{SC flip} decoder. Our algorithm preserves the low memory requirements of the basic SC decoder and adjusts the required decoding effort to the signal quality. In the waterfall region, its average computational complexity is almost as low as that of the SC decoder.
\end{abstract}


%
\IEEEpeerreviewmaketitle

\section{Introduction}\label{sec:introduction}
Polar codes \cite{Arikan2009} are particularly attractive from a theoretical point of view because they are the first codes that are both highly structured and provably optimal for a wide range of applications (in the sense of optimality that pertains to each application). Moreover, they can be decoded using an elegant, albeit suboptimal, successive cancellation (SC) algorithm, which has computational complexity $O(N \log N)$~\cite{Arikan2009}, where $N = 2^n,~n \in \mathbb{Z},$ is the blocklength of the code, and memory complexity $O(N)$~\cite{Leroux2011}. Even though the SC decoder is suboptimal, it is sufficient to prove that polar codes are capacity achieving in the limit of infinite blocklength.

Unfortunately, the error correcting performance of SC decoding at finite blocklengths is not as good as that of other modern codes, such as LDPC codes. To improve the finite blocklength performance, more sophisticated algorithms, such as \emph{SC list} decoding~\cite{Tal2011} and \emph{SC stack} decoding~\cite{Chen2013}, were introduced recently. These algorithms use SC as the underlying decoder, but improve its performance by exploring multiple paths on a decision tree simultaneously, with each path resulting in one candidate codeword. The computational and memory complexities of SC list decoding are $O(LN\log N)$ and $O(LN),$ respectively, where $L$ is the \emph{list size} parameter, whereas the computational and memory complexities of SC stack decoding are $O(DN\log N)$ and $O(DN),$ respectively, where $D$ is the \emph{stack depth} parameter.

Since an exhaustive search through all paths is prohibitively complex, choosing a suitable strategy for pruning unlikely paths is an important ingredient for low-complexity tree search algorithms. To this end, in \cite{Chen2013}, some path pruning-based methods were proposed in order to reduce the computational complexity of both SC stack and SC list decoding. An alternative approach to reduce the computational complexity of SC list decoding was taken in~\cite{Li2012,Sarkis2014}, where decoding starts with list size $1$, and the list size is increased only when decoding fails (failures are detected using a CRC), up to the maximum list size $L$. Moreover, in~\cite{Cao2013} SC list decoding is employed only for the least reliable bits of the polar code, thus also reducing the computational complexity. However, in~\cite{Cao2013} $L$ distinct paths are still followed in parallel.

Unfortunately, when implementing any decoder in hardware, one always has to provision for the worst case in terms of hardware resources. For the reduced-complexity SC list decoders in \cite{Chen2013,Sarkis2014,Li2012,Cao2013} and the reduced-complexity SC stack decoder in~\cite{Chen2013} this means that $O(LN)$ and $O(DN)$ memory needs to be instantiated, respectively. Moreover, the reduced-complexity list SC and stack SC algorithms also have a significantly higher computational complexity than that of the original SC algorithm. 

\subsubsection*{Contribution} In this paper, we describe a new SC-based decoding algorithm, called \emph{SC flip}, which retains the $O(N)$ memory complexity of the original SC algorithm and has an average computational complexity that is practically $O(N\log N)$ at high SNR, while still providing a significant gain in terms of error correcting performance.

\section{Polar Codes and Successive Cancellation Decoding}\label{sec:polarsc}

\subsection{Construction of Polar Codes}
Let $W$ denote a binary input memoryless channel with input $u~\in~\{0,1\}$, output $y~\in~\mathcal{Y}$, and transition probabilities $W(y|u)$. A polar code is constructed by recursively applying a $2 \times 2$ \emph{channel combining} transformation on $2^n$ independent copies of $W$, followed by a \emph{channel splitting} step~\cite{Arikan2009}. This results in a set of $N = 2^n$ synthetic channels, denoted by $W_n^{(i)}(y_1^N,u_1^{i-1}|u_i),~i=1,\hdots,N$. Let $Z_i \triangleq Z\left(W_n^{(i)}(Y_1^N,U_ 1^{i-1}|U_i)\right),~i=1,\hdots,N$, where $Z(W)$ is the Bhattacharyya parameter of $W$, which can be calculated using various methods (cf. \cite{Arikan2009,Pedarsani2011,Tal2013}). The construction of a polar code of rate $R \triangleq \frac{k}{N},~0 < k < N,$ is completed by choosing the $k$ best synthetic channels (i.e., the synthetic channels with the lowest $Z_i$) as \emph{non-frozen} channels which carry information bits, while \emph{freezing} the input of the remaining channels to some values $u_i$ that are known both to the transmitter and to the receiver. The set of frozen channel indices is denoted by $\mathcal{A}^c$ and the set of non-frozen channel indices is denoted by $\mathcal{A}$. The encoder generates a vector $u_1^N$ by setting $u_{\mathcal{A}^c}$ equal to the known frozen values, while choosing $u_{\mathcal{A}}$ freely. A codeword is obtained as $x_1^N = u_{1}^NG_N,$ where $G_N$ is the generator matrix~\cite{Arikan2009}.


\subsection{Successive Cancellation Decoding}
The SC decoding algorithm~\cite{Arikan2009} starts by computing an estimate of $u_1$, denoted by $\hat{u}_1$, based only on the received values $y_1^N$. Subsequently, $u_2$ is estimated using  $(y_1^N,\hat{u}_1),$ etc. Since $u_i,~i\in \mathcal{A}^c$ are known to the receiver, the real task of SC decoding is to estimate $u_i,~i\in \mathcal{A}$. Let the log-likelihood ratio (LLR) for $W_n^{(i)}(y_1^N,\hat{u}_1^{i-1}|u_i)$ be defined as 
\begin{align}
	L^{(i)}_n(y_1^N,\hat{u}_1^{i-1}|u_i) \triangleq \log \frac{W_n^{(i)}(y_1^N,\hat{u}_1^{i-1}|u_i=0)}{W_n^{(i)}(y_1^N,\hat{u}_1^{i-1}|u_i=1)}. 
\end{align}
Decisions are taken according to
\begin{align}
	\hat{u}_i & =\left\{ \begin{matrix} 0, & L^{(i)}_n(y_1^N,\hat{u}_1^{i-1}|u_i) \geq 0 \text{ and } i \in \mathcal{A}, \\ 1, & L^{(i)}_n(y_1^N,\hat{u}_1^{i-1}|u_i) < 0 \text{ and } i \in \mathcal{A}, \\u_i, & i \in \mathcal{A}^c. \end{matrix} \right. \label{eq:scdec}
\end{align}
The decision LLRs $L^{(i)}_n(y_1^N,\hat{u}_1^{i-1}|u_i)$ can be calculated efficiently through a computation graph which contains two types of nodes, namely $f$ nodes and $g$ nodes. An example of this graph for $N = 8$ is given in Fig.~\ref{fig:scbutterfly}. Both types of nodes have two input LLRs, denoted by $L_1$ and $L_2$, and one output LLR, denoted by $L$. The $g$ nodes have an additional input called the \emph{partial sum}, denoted by $u$. The partial sums form the \emph{decision feedback} part of the SC decoder. The min-sum update rules~\cite{Leroux2011} for the two types of nodes are
\begin{align}
	f(L_1,L_2)						& = \text{sign}(L_1)\text{sign}(L_2)\min \left(|L_1|,|L_2|\right), \\
	g(L_1,L_2,u)	& = (-1)^{u}L_1 + L_2.
\end{align} 
The partial sums at stage $(s-1)$ can be calculated from the partial sums at stage $s,~s \in \{1,\hdots,n\},$ as
\begin{align}
	\uu{s-1}{2i - 1 - [(i-1) \bmod 2^{s-1}]} 	& = \uu{s}{2i-1} \oplus \uu{s}{2i}, \\
	\uu{s-1}{2^{s-1} + 2i-1 - [(i-1) \bmod 2^{s-1}]} 		& =  \uu{s}{2i},
\end{align}
where
\begin{equation}
    \uu{n}{i} \triangleq \hat{u}_i, \qquad \forall i \in \left\{1,\hdots,N\right\}.
\end{equation}

The computation graph contains $N \log (N+1)$ nodes and each node only needs to be activated once. Thus, the computational complexity of SC decoding is $O(N \log N)$. A straightforward implementation of the computation graph in Fig.~\ref{fig:scbutterfly} requires $O(N \log N)$ memory positions. However, by cleverly re-using memory locations, it is possible to reduce the memory complexity to $O(N)$~\cite{Leroux2011}.

\begin{figure}
  \centering
  \scalebox{0.6}{\begin{tikzpicture}[x=1cm,y=-1cm,
    >=stealth',
  block/.style={rounded corners,draw=black,fill=blue!30},
  dblock/.style={draw=black},
  fplus/.style={dashed,thick,<-,draw=red},
  fminus/.style={thick,<-,} ]
  \node[block,fill=green!30] at (9,0) (l00) {$\quad \, \LL{0}{1} \quad \, $};
  \node[block,fill=green!30] at (9,1) (l01) {$\quad \, \LL{0}{2} \quad \, $};
  \node[block,fill=green!30] at (9,2) (l02) {$\quad \, \LL{0}{3} \quad \, $};
  \node[block,fill=green!30] at (9,3) (l03) {$\quad \, \LL{0}{4} \quad \, $};
  \node[block,fill=green!30] at (9,4) (l04) {$\quad \, \LL{0}{5} \quad \, $};
  \node[block,fill=green!30] at (9,5) (l05) {$\quad \, \LL{0}{6} \quad \, $};
  \node[block,fill=green!30] at (9,6) (l06) {$\quad \, \LL{0}{7} \quad \, $};
  \node[block,fill=green!30] at (9,7) (l07) {$\quad \, \LL{0}{8} \quad \, $};
  
  \node[block,fill=green!30] at (6,0) (l10) {$\quad \, \LL{1}{1} \quad \, $};
  \node[block] at (6,1) (l11) {$\LL{1}{2} (\uu{1}{1})$};
  \node[block,fill=green!30] at (6,2) (l12) {$\quad \, \LL{1}{3} \quad \, $};
  \node[block] at (6,3) (l13) {$\LL{1}{4} (\uu{1}{3})$};
  \node[block,fill=green!30] at (6,4) (l14) {$\quad \, \LL{1}{5} \quad \, $};
  \node[block] at (6,5) (l15) {$\LL{1}{6} (\uu{1}{5})$};
  \node[block,fill=green!30] at (6,6) (l16) {$\quad \, \LL{1}{7} \quad \, $};
  \node[block] at (6,7) (l17) {$\LL{1}{8} (\uu{1}{7})$};
  
  \node[block,fill=green!30] at (3,0) (l20) {$\quad \, \LL{2}{1} \quad \, $};
  \node[block] at (3,1) (l21) {$\LL{2}{2} (\uu{2}{1})$};
  \node[block,fill=green!30] at (3,2) (l22) {$\quad \, \LL{2}{3} \quad \, $};
  \node[block] at (3,3) (l23) {$\LL{2}{4} (\uu{2}{3})$};
  \node[block,fill=green!30] at (3,4) (l24) {$\quad \, \LL{2}{5} \quad \, $};
  \node[block] at (3,5) (l25) {$\LL{2}{6} (\uu{2}{5})$};
  \node[block,fill=green!30] at (3,6) (l26) {$\quad \, \LL{2}{7} \quad \, $};
  \node[block] at (3,7) (l27) {$\LL{2}{8} (\uu{2}{7})$};

  \node[block,fill=green!30] at (0,0) (l30) {$\quad \, \LL{3}{1} \quad\, $};
  \node[block,fill=blue!30] at (0,1) (l31) {$\LL{3}{2} (\uu{3}{1})$};
  \node[block,fill=green!30] at (0,2) (l32) {$\quad \, \LL{3}{3}\quad\, $};
  \node[block,fill=blue!30] at (0,3) (l33) {$\LL{3}{4} (\uu{3}{3})$};
  \node[block,fill=green!30] at (0,4) (l34) {$\quad \, \LL{3}{5}\quad\, $};
  \node[block,fill=blue!30] at (0,5) (l35) {$\LL{3}{6} (\uu{3}{5})$};
  \node[block,fill=green!30] at (0,6) (l36) {$\quad \, \LL{3}{7}\quad\, $};
  \node[block,fill=blue!30] at (0,7) (l37) {$\LL{3}{8} (\uu{3}{7})$};

  \node[dblock,fill=black!30] at (-2,0) (d0) {$\hat{u}_{1}$};
  \node[dblock,fill=black!30] at (-2,1) (d1) {$\hat{u}_{2}$};
  \node[dblock,fill=black!30] at (-2,2) (d2) {$\hat{u}_{3}$};
  \node[dblock,fill=black!30] at (-2,3) (d3) {$\hat{u}_{4}$};
  \node[dblock,fill=black!30] at (-2,4) (d4) {$\hat{u}_{5}$};
  \node[dblock,fill=black!30] at (-2,5) (d5) {$\hat{u}_{6}$};
  \node[dblock,fill=black!30] at (-2,6) (d6) {$\hat{u}_{7}$};
  \node[dblock,fill=black!30] at (-2,7) (d7) {$\hat{u}_{8}$};

	\draw[fminus] (d0.east) -- (l30.west);
	\draw[fminus] (d0.east) -- (l30.west);
	\draw[fminus] (d1.east) -- (l31.west);
	\draw[fminus] (d2.east) -- (l32.west);
	\draw[fminus] (d3.east) -- (l33.west);
	\draw[fminus] (d4.east) -- (l34.west);
	\draw[fminus] (d5.east) -- (l35.west);
	\draw[fminus] (d6.east) -- (l36.west);
	\draw[fminus] (d7.east) -- (l37.west);

  \draw[fminus] (l30.east) -- (l20.west);
  \draw[fminus] (l30.east) -- (l24.west);
  %
  \draw[fminus] (l20.east) -- (l10.west);
  \draw[fminus] (l20.east) -- (l12.west);
  \draw[fminus] (l24.east) -- (l14.west);
  \draw[fminus] (l24.east) -- (l16.west);
  \draw[fminus] (l10.east) -- (l00.west);
  \draw[fminus] (l10.east) -- (l01.west);
  \draw[fminus] (l12.east) -- (l02.west);
  \draw[fminus] (l12.east) -- (l03.west);
  \draw[fminus] (l14.east) -- (l04.west);
  \draw[fminus] (l14.east) -- (l05.west);
  \draw[fminus] (l16.east) -- (l06.west);
  \draw[fminus] (l16.east) -- (l07.west);
  \draw[fminus] (l31.east) -- (l20.west);
  \draw[fminus] (l31.east) -- (l24.west);
  \draw[fminus] (l32.east) -- (l21.west);
  \draw[fminus] (l32.east) -- (l25.west);
  \draw[fminus] (l21.east) -- (l10.west);
  \draw[fminus] (l21.east) -- (l12.west);
  \draw[fminus] (l25.east) -- (l14.west);
  \draw[fminus] (l25.east) -- (l16.west);
  \draw[fminus] (l33.east) -- (l21.west);
  \draw[fminus] (l33.east) -- (l25.west);
  \draw[fminus] (l34.east) -- (l22.west);
  \draw[fminus] (l34.east) -- (l26.west);
  \draw[fminus] (l22.east) -- (l11.west);
  \draw[fminus] (l22.east) -- (l13.west);
  \draw[fminus] (l26.east) -- (l15.west);
  \draw[fminus] (l26.east) -- (l17.west);
  \draw[fminus] (l11.east) -- (l00.west);
  \draw[fminus] (l11.east) -- (l01.west);
  \draw[fminus] (l13.east) -- (l02.west);
  \draw[fminus] (l13.east) -- (l03.west);
  \draw[fminus] (l15.east) -- (l04.west);
  \draw[fminus] (l15.east) -- (l05.west);
  \draw[fminus] (l17.east) -- (l06.west);
  \draw[fminus] (l17.east) -- (l07.west);
  \draw[fminus] (l35.east) -- (l22.west);
  \draw[fminus] (l35.east) -- (l26.west);
  \draw[fminus] (l36.east) -- (l23.west);
  \draw[fminus] (l36.east) -- (l27.west);
  \draw[fminus] (l23.east) -- (l11.west);
  \draw[fminus] (l23.east) -- (l13.west);
  \draw[fminus] (l27.east) -- (l15.west);
  \draw[fminus] (l27.east) -- (l17.west);
  \draw[fminus] (l37.east) -- (l23.west);
  \draw[fminus] (l37.east) -- (l27.west);

  \draw[decorate,decoration={brace,amplitude=10pt}] (10,0) -- (10,7)
  node[black,midway,rotate=270,yshift=2em]{Channel LLRs -- stage $s=0$};
  
\end{tikzpicture}}
  \caption{The computation graph of the SC decoder for $N=8$. The $f$ nodes are green and $g$ nodes are blue and in the parentheses are the partial sums that are used by each $g$ node.}
  \label{fig:scbutterfly}
\end{figure}
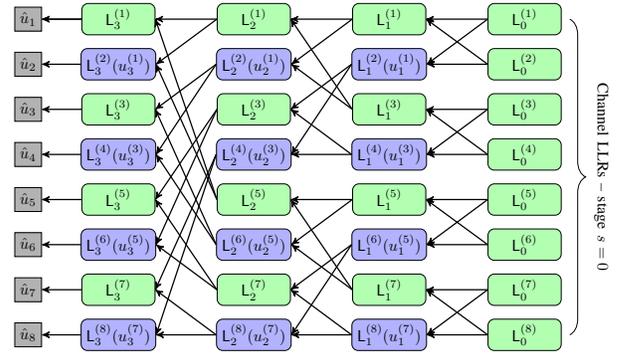

\section{Error Propagation in SC Decoding}\label{sec:errorprop}
In SC decoding, erroneous bit decisions can be caused by channel noise \emph{or} by error propagation due to previous erroneous bit decisions. The first erroneous decision is always caused by the channel noise since there are no previous errors, so error propagation does not affect the frame error rate of polar codes, but only the bit error rate. 


\subsection{Effect of Error Propagation}
The erroneous decisions due to error propagation are caused by erroneous decision feedback, which in turns leads to erroneous partial sums.  Erroneous partial sums can corrupt the output LLR values at all stages, including, most importantly, the decision LLRs at level $n$. 

\begin{figure}[t]
	\centering
	\includegraphics[width=0.42\textwidth]{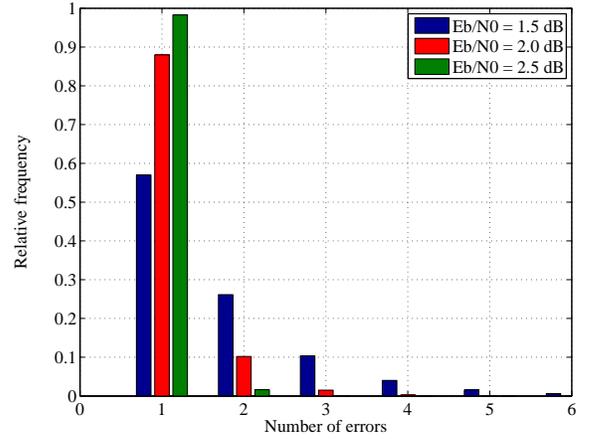}
	\caption{Histogram showing the relative frequency of the number of errors caused by the channel for a polar code with $N = 1024$ and $R = 0.5$ for three different SNR values.}\label{fig:ErrorsSNR}
\end{figure}

\begin{figure}[t]
	\centering
	\includegraphics[width=0.42\textwidth]{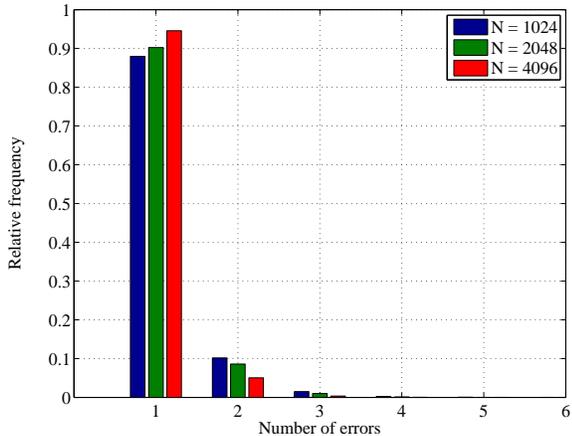}
	\caption{Histogram showing the relative frequency of the number of errors actually caused by the channel for Eb/N0 = 2.00 and three different codelengths, $ N = 1024, 2048, 4096 $.}\label{fig:ErrorsN}
	\vspace{-0.2cm}
\end{figure}

For example, assume that, for the polar code in Fig.~\ref{fig:scbutterfly}, the frozen set is $\mathcal{A}^c = \{1,2,5,6\}$ and the information set is $\mathcal{A} = \{3,4,7,8\}$. Moreover, assume that the all-zero codeword was transmitted and that $\hat{u}_3$ was erroneously decoded as $\hat{u}_3 = 1$ due to channel noise. Now suppose that the two LLRs that are used to calculate the next decision LLR (i.e., $\LL{3}{4}$), namely, $\LL{2}{2}$ and $\LL{2}{6}$, are both positive and $\LL{2}{2} > \LL{2}{6}$. By applying the $g$ node update rule with $u = \uu{3}{3} = \hat{u}_3 = 1$, the resulting decision LLR $\LL{3}{4} = \LL{2}{6}-\LL{2}{2}$ has a negative value which leads to a second erroneous decision, while with the correct partial sum $u = 0$ the decision would have been correct.


\subsection{Significance of Error Propagation}
The foregoing analysis of the effects of error propagation insinuates the following question: Given that we had an erroneously decoded codeword with many erroneous bits, how many of these bits were actually wrong because of channel noise rather than due to previous erroneous decisions? In order to answer to this question, we employ an oracle-assisted SC decoder. Each time an error occurs at the decision level, the oracle corrects it instantaneously without allowing it to affect any future bit decisions. Moreover, the oracle-assisted SC decoder counts the number of times it had to correct an erroneous decision.

In Fig.~\ref{fig:ErrorsSNR} we plot a histogram of the number of errors caused by channel noise (given that there was at least one error) for three different Eb/N0 values for a polar code with $N = 1024$ and $R = 0.5$ over an AWGN channel. We observe that most frequently the channel introduces only one error and that this behavior becomes even more prominent for increasing Eb/N0 values. In Fig.~\ref{fig:ErrorsN} we plot a histogram of the number of errors caused by channel noise for polar codes with three different blocklengths and $R = 0.5$ over an AWGN channel at $\text{Eb/N0}=2~\text{dB}$. We observe that, the relative frequency of the single error event increases with increasing blocklengths. This happens because, as $N$ gets larger, the synthetic channels $W_n^{(i)}(y_1^N,u_1^{i-1}|u_i)$ become more polarized, meaning that all information channels in $\calA$ become better.

\subsection{Oracle-Assisted SC Decoder}\label{sec:scoracle}
From the discussion in the previous section, it is clear that, by identifying the position of the first erroneous bit decision and inverting that decision, the performance of the SC decoder could be improved significantly. In order to examine the potential benefits of correcting a single error we employ a second oracle-assisted SC decoder, which is only allowed to intervene \emph{once} in the decoding process in order to correct the first erroneous bit decision.

In Fig.~\ref{fig:SC_Oracle_3N} we compare the performance of the SC decoder with that of the oracle-assisted SC decoder for a polar code of three blocklengths and $R = 0.5$ over an AWGN channel. We observe that correcting a single erroneous bit decision significantly improves the performance of the SC decoder.
\begin{figure}[t]
	\centering
	\includegraphics[width=0.42\textwidth]{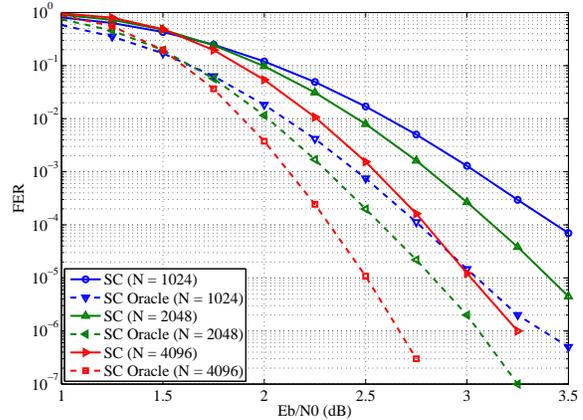}
	\caption{Performance of oracle-assisted SC decoder compared to the SC decoder for $ N = 1024, 2048, 4096 $ and $R=0.5$.}\label{fig:SC_Oracle_3N}
	\vspace{-0.2cm}
\end{figure}


\section{SC Flip Decoding}\label{sec:scflip}

The goal of SC flip decoding is to identify the first error that occurs during SC decoding without the aid of an oracle.

\subsection{SC Flip Decoding Algorithm}
Assume that we are given a polar code of rate $\tilde{R} = \frac{k}{N}$ with a set of information bits $\tilde{\mathcal{A}}$. We use an $r$-bit CRC that tells us, with high probability, whether the codeword estimate $\hat{u}_1^N$ given by the SC decoder is a valid codeword or not.  In order to incorporate the CRC, the rate of the polar code is increased to $R = \tilde{R} + \frac{r}{N} = \frac{k+r}{N}$, so that the effective information rate remains unaltered. Equivalently, the set of information bits $\tilde{\mathcal{A}}$ is extended with the $r$ most reliable channel indices in $\tilde{\mathcal{A}}^\text{c}$, denoted by $\tilde{\mathcal{A}}^\text{c}_{\text{$r$--$\max$}}$ Thus, $\mathcal{A} = \tilde{\mathcal{A}} \cup \tilde{\mathcal{A}}^\text{c}_{\text{$r$--$\max$}}$.

The SC flip decoder starts by performing standard SC decoding in order to produce a first estimated codeword $\hat{u}_1^N$. If $\hat{u}_1^N$ passes the CRC, then decoding is completed. If the CRC fails, the SC flip algorithm is given $T$ additional attempts to identify the first error that occurred in the codeword. To this end, let $\mathcal{U}$ denote the set of the $T$ least reliable decisions, i.e., the set containing the indices $i \in \mathcal{A}$ corresponding to the $T$ smallest $|L_n^{(i)}(y_1^N,\hat{u}_1^{i-1}|u_i)|$ values. After the set $\mathcal{U}$ has been constructed, SC decoding is restarted for a total of no more than $T$ additional attempts. In each attempt, a single $\hat{u}_k, k \in \mathcal{U},$ is flipped with respect to the initial decision of the SC algorithm. The algorithm terminates when a valid codeword has been found or when all $T$ additional attempts have failed. Note that, for $T = 0$, SC flip decoding is equivalent to SC decoding. 

The SC flip algorithm is formalized in the \textsc{SCFlip}$(y_1^N,\mathcal{A},k)$ function in Fig.~\ref{fig:scflip}. The \textsc{SC}$(y_1^N,\mathcal{A},k)$ function performs SC decoding based on the channel output $y_1^N$ and the set of non-frozen bits $\mathcal{A}$ with a slight twist: when $k > 0$, the codeword bit $u_k$ is decoded by flipping the value obtained from the decoding rule~\eqref{eq:scdec}. 

Note that SC flip decoding is similar to chase decoding for polar codes~\cite{Sarkis2013}. The main differences are that SC flip decoding only considers error patterns containing a single error and that these error patterns are not generated offline using the a-priori reliabilities of the synthetic channels $W_n^{(i)}(y_1^N,u_1^{i-1}|u_i)$, but online using the decision LLRs $\LL{n}{i}$, which reflect the actual reliabilities of the bit decisions for each transmitted codeword and channel noise realization.

\begin{figure}
\centering
\algsetup{indent=1.5em}
\begin{algorithmic}[1]
	\STATE \textbf{function} \textsc{SCFlip}$(y_1^N,\mathcal{A},T)$
	\STATE $\left(\hat{u}_1^N,L(y_1^N,\hat{u}_1^{i-1}|u_i)\right) \leftarrow \textsc{SC}(y_1^N,\mathcal{A},0)$; \label{alg:SCinit}
	\IF{$T > 1$ and \textsc{CRC}$(\hat{u}_1^N) = $ failure}  \label{alg:CRCinit}
		\STATE $\mathcal{U} \leftarrow$ $i \in \mathcal{A}$ of $T$ smallest $|L(y_1^N,\hat{u}_1^{i-1}|u_i)|$; \label{alg:sorting}
		\FOR{$j \leftarrow 1$ \TO $T$} \label{alg:forbegin}
			\STATE $k \leftarrow \mathcal{U}(j)$;
			\STATE $\hat{u}_1^N \leftarrow \textsc{SC}(y_1^N,\mathcal{A},k)$; \label{alg:SCrestart}
			\IF{\textsc{CRC}$(\hat{u}_1^N) = $ success} \label{alg:CRCrestart}
				\STATE break;
			\ENDIF
		\ENDFOR \label{alg:forend}
	\ENDIF
	\RETURN $\hat{u}_1^N$;
\end{algorithmic}
\caption{SC flip decoding with maximum trials $T$.}\label{fig:scflip}
\end{figure}

\subsection{Complexity of SC Flip Decoding}\label{sec:complexity}
In this section, we derive the worst-case and average-case computational complexities of the SC flip algorithm, as well as its memory complexity.

\begin{proposition}
The worst-case computational complexity of the \textsc{SCFlip} algorithm defined in Fig.~\ref{fig:scflip} is $O(TN\log N)$. \label{prop:worstcase}
\end{proposition}
\begin{proof}
SC decoding in line~\ref{alg:SCinit} has complexity $O(N \log N)$ and the computation of the CRC in line~\ref{alg:CRCinit} has complexity $O(N)$. Moreover, the sorting step in line~\ref{alg:sorting} can be implemented with complexity $O(N \log N)$ (e.g., using merge sort). Finally, the operations in the loop (lines~\ref{alg:forbegin}--\ref{alg:forend}) have complexity $O(N \log N)$ and the loop runs $T$ times in the worst case. Thus, the overall worst-case complexity is $O(TN \log N)$.
\end{proof}
Proposition~\ref{prop:worstcase} shows that, in the worst case, the complexity of our algorithm increases linearly with the parameter $T$, meaning that its complexity scaling is no better than that of SC list decoding. However, if we consider the \emph{average} complexity, then the situation is much more favorable, as the following result shows.

\begin{proposition}\label{prop:avgcompl}
Let $P_e(R,\text{SNR})$ denote the frame error rate of a polar code of rate $R$ at the given SNR point. Then, the average-case computational complexity of the \textsc{SCFlip} algorithm defined in Fig.~\ref{fig:scflip} is $O(N\log N(1+T \cdot P_e(R,\text{SNR})))$, where $R = \frac{k+r}{N}$. \label{prop:averagecase}
\end{proposition}
\begin{proof}
It suffices to observe that the loop in lines~\ref{alg:forbegin}--\ref{alg:forend} runs only if SC decoding fails and the CRC detects the failure, which happens with probability at most $P_e(R,\text{SNR})$.
\end{proof}
As the SNR increases, the FER drops asymptotically to zero. Thus, for high SNR the average computational complexity of SC flip decoding converges to the computational complexity of SC decoding. In other words, SC flip exhibits an \emph{energy-proportional} behavior where more energy is spent when the problem is difficult (i.e., at low SNR) and less energy is spent when the problem is easy (i.e., at high SNR).

\begin{proposition}
The \textsc{SCFlip} algorithm defined in Fig.~\ref{fig:scflip} requires $O(N)$ memory positions.
\end{proposition}
\begin{proof}
SC decoding in line~\ref{alg:SCinit} requires $O(N)$ memory positions. The storage of the CRC calculated in lines~\ref{alg:CRCinit} and \ref{alg:CRCrestart} requires exactly $C$ memory positions, where $C \leq N$. The sorting step in line~\ref{alg:sorting} can be implemented with $O(N)$ memory positions (e.g., using merge sort), while storing the $T$ smallest values requires exactly $T$ memory positions, where $T \leq N$. Moreover, the SC decoding performed in line~\ref{alg:SCrestart} can re-use the memory positions of the SC decoding in line~\ref{alg:SCinit}, so no additional memory is required. Thus, the overall memory scaling behavior is $O(N)$.
\end{proof}

\subsection{Error Correcting Performance of SC Flip Decoding}

\begin{figure}[t]
	\centering
	\includegraphics[width=0.42\textwidth]{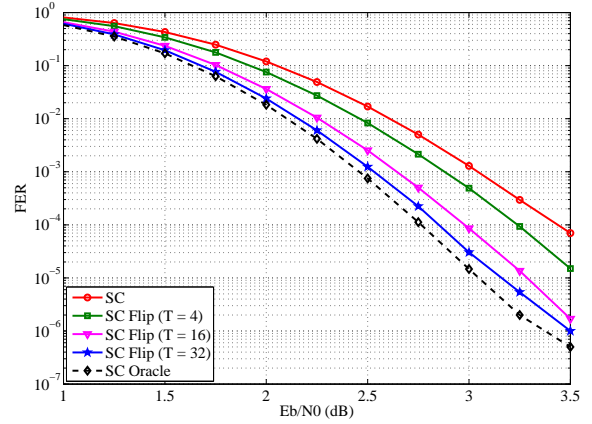}
	\caption{Frame error rate of SC decoding, SC flip decoding with $T=4,16, 32$ and the oracle-assisted SC decoder for a polar code of length $N=1024$ and $R=0.5$.}\label{fig:Flip1024}
\end{figure}

\begin{figure}[t]
	\centering
	\includegraphics[width=0.42\textwidth]{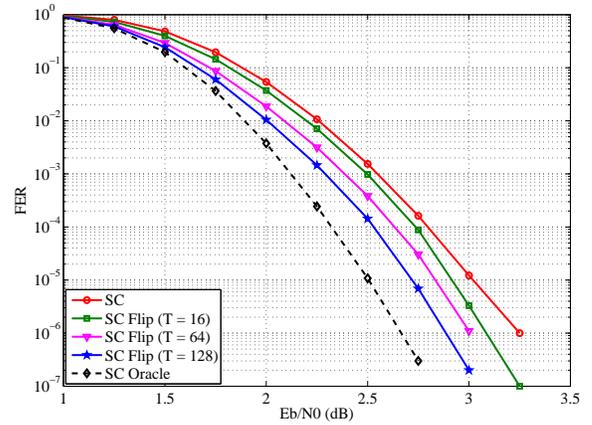}
	\caption{Frame error rate of SC decoding, SC flip decoding with $T=4,16, 32$ and the oracle-assisted SC decoder for a polar code of length $N=4096$ and $R=0.5$.}\label{fig:Flip4096}
	\vspace{-0.3cm}
\end{figure}

In Fig.~\ref{fig:Flip1024} we compare the performance of the SC flip decoder with $T=4,16,32,$ and a 16-bit CRC with the SC decoder and the oracle-assisted SC decoder described in Section~\ref{sec:scoracle}. Note that the oracle-assisted decoder characterizes a performance bound for the SC flip decoder. We observe that SC flip decoding with $T=4$ already leads to a gain of one order of magnitude in terms of FER at Eb/N0 = 3.5 dB. With $T=32$, we can reap all the benefits of the oracle-assisted SC decoder, since the $T=32$ curve is shifted to the right with respect to the oracle-assisted curve by an amount that corresponds exactly to the rate loss incurred by the $16$-bit CRC.

In Fig.~\ref{fig:Flip4096} we depict the same curves for a codelength $N = 4096$, while keeping the ratio $\frac{T}{N}$ constant. We observe that it seems to become more difficult to reach the bound performance of the oracle-assisted SC decoder. As $N$ increases, the channels get more polarized, which would suggest the opposite behavior. However, at the same time, the absolute number of the possible positions for the first error increases as well. Our results suggest that the aforementioned negative effect negates the positive effect of channel polarization.

In Fig.~\ref{fig:FlipList}, we compare the performance of standard SC decoding, SC flip decoding, and SC list decoding. 
We observe that the performance of the SC flip decoder with $T = 32$ is almost identical to that of the SC list decoder with $L = 2$, but with half the computational complexity at high Eb/N0 values and half the memory complexity at all Eb/N0 values. For higher list sizes, such as $L=4$, SC list decoding outperforms SC flip decoding, at the cost of significantly higher complexity, since the performance of SC flip decoding is limited by the fact that it can only correct a single error.

\subsection{Average Computational Complexity of SC Flip Decoding}

In Fig. \ref{fig:complexity}, we compare the average computational complexity of standard SC decoding, SC list decoding, and SC flip decoding. We observe that, as predicted by Proposition~\ref{prop:avgcompl}, at low SNR the average computational complexity of SC flip decoding is $(T+1)$ times larger than that of SC decoding but at higher SNR the computational complexity is practically identical to that of SC decoding. Moreover, the energy-proportional behavior of SC flip decoding is evident since, contrary to SC list decoding, the computational complexity decreases rapidly with decreasing difficulty of the decoding problem (i.e., increasing SNR). We also emphasize that SC flip decoding is not a viable option for the low SNR region, but this a not a region of interest for practical systems because the FER is very high.

\begin{figure}[t]
	\centering
	\includegraphics[width=0.42\textwidth]{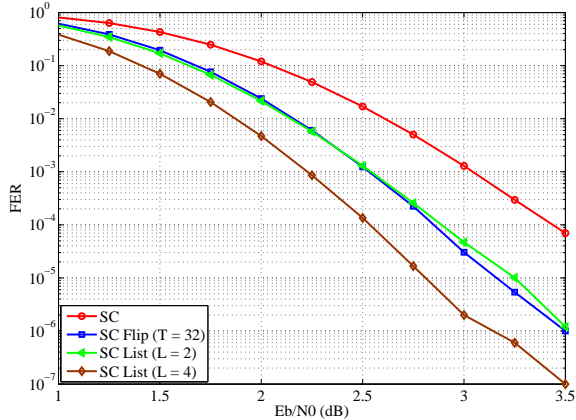}
	\caption{Frame error rate of SC decoding, SC flip decoding with $T=32$ and SC list decoding with $L=2, 4$ for a polar code of length $N=1024$ and $R=0.5$.}\label{fig:FlipList}
\end{figure}


\begin{figure}[t]
	\centering
	\includegraphics[width=0.42\textwidth]{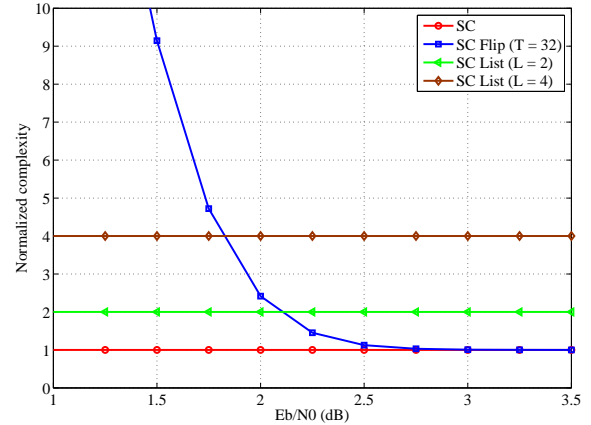}
	\caption{Average complexity of SC flip decoding normalized with respect to the complexity of SC decoding for a polar code of length $N=1024$ and $R=0.5$.}\label{fig:complexity}
\end{figure}

\section{Conclusion}\label{sec:concl}
In this paper we have introduced \emph{successive cancellation flip} decoding for polar codes. This algorithm improves the frame error rate performance by opportunistically retrying alternative decisions for bits that turned out to be unreliable in a failing initial decoding iteration. By exploring alternative passes in the decoding tree one after another until a correct codeword is found, the average complexity and memory requirements are kept low, while approaching the performance of more complex tree-search based decoders.

%
%
%
%
\ifCLASSOPTIONcaptionsoff
  \newpage
\fi

\end{document}